\documentclass[a4paper, 12pt]{article}

\usepackage[latin1]{inputenc}
\usepackage[T1]{fontenc}
\usepackage{graphicx}
\usepackage{caption}
\usepackage{subcaption}
\usepackage[english]{babel}
\usepackage[top=2.5cm, bottom=2.5cm, left=3cm, right=3cm]{geometry} 
\usepackage{amsmath}
\usepackage[babel=true]{csquotes}
\usepackage{harvard}
\usepackage{units} 
\usepackage{amsthm}
\usepackage{amsfonts}
\usepackage{setspace}
\usepackage{mathtools}
\usepackage{xcolor}

\newtheorem{theorem}{Theorem}[section]
\newtheorem{proposition}{Proposition}[section]

\newtheorem{remark}{Remark}

\newcommand{\hs}{\hspace{0.3em}}

\graphicspath{{images/}}

\title{Mean Field Game Approach to Bitcoin Mining}

\author{
Charles Bertucci\footnote{CMAP, Ecole Polytechnique, UMR 7641, 91120 Palaiseau, France}
\and
Louis Bertucci\footnote{Institut Louis Bachelier, 75002 Paris, France} \footnote{Haas School of Business, UC Berkeley, Berkeley, California}
\and
Jean-Michel Lasry\footnote{Université Paris-Dauphine, PSL Research University,UMR 7534, CEREMADE, 75016 Paris, France}
\and
Pierre-Louis Lions\footnote{Collège de France, 3 rue d'Ulm, 75005, Paris, France} \hspace{0.5ex}\footnotemark[4]}

\begin{document}
\maketitle


\vspace{1cm}

\begin{abstract}

We present an analysis of the Proof-of-Work consensus algorithm, used on the Bitcoin blockchain, using a Mean Field Game framework. Using a master equation, we provide an equilibrium characterization of the total computational power devoted to mining the blockchain (hashrate). From a simple setting we show how the master equation approach allows us to enrich the model by relaxing most of the simplifying assumptions. The essential structure of the game is preserved across all the enrichments. In deterministic settings, the hashrate ultimately reaches a steady state in which it increases at the rate of technological progress. In stochastic settings, there exists a target for the hashrate for every possible random state. As a consequence, we show that in equilibrium the security of the underlying blockchain is either $i)$ constant, or $ii)$ increases with the demand for the underlying cryptocurrency.

\end{abstract}

\vfill

\pagebreak

\section{Introduction}

Blockchains are attracting more and more interest from different areas of research. Understanding the game theoretic aspects of a blockchain is fundamental as it was built to operate in a decentralized setup, in which agents do not face any strategy restrictions whatsoever. This should be understood as no restrictions beside the basic protocol. For instance an agent could send random messages or disconnect for a while, there is no way to enforce \enquote{good behavior} except with a properly designed incentive scheme.

This paper is devoted to modeling the competition between miners in a \emph{Proof-of-Work} (PoW) based blockchain using a mean field game (MFG) framework. We derive a simple model which relies on the so-called master equation. Our model can be easily adapted to treat a various set of situations encountered in the mining economy. In particular we show that the mining game resulting from the PoW algorithm is well posed and contains some intrinsic stability that translates to security.

A blockchain is a communication protocol allowing a group of users to update and maintain a distributed database (or ledger) in a trustless fashion. Every participant holds his own version of the ledger, and the blockchain protocol allows them to keep all versions synchronised. A blockchain allows consensus to arise without relying on any trusted agents. At the fundamental level, entries in a blockchain can be of any kinds, it is just a list of items (transactions, assets ownership, etc.), but there needs to be a native currency (called cryptocurrency) used as a reward mechanism. Indeed, maintainers are rewarded with the native cryptocurrency for validating new entries. PoW is the consensus algorithm used on the biggest permissionless blockchains including Bitcoin. In other words, it is the algorithm that allows every honest participant to update their version of the ledger in sync.

The PoW algorithm can be thought of as a non-cooperative game with a very large number of players. Let's now briefly describe the PoW algorithm (for more details see \citeasnoun{nakamoto2008bitcoin} and \citeasnoun{garay2015bitcoin}). Essentially, miners are competing to brute force the solution of a hash-based puzzle. The winner gets to propose a new block that would be accepted by other miners if valid, in which case it determines the next puzzle to be solved. The winning miner gets rewarded in the native cryptocurrency when he finds a block. Since this problem can only be solved by brute-force, more computing power will speed up the resolution. The computing power devoted to mining is also called the hashrate, which represents the number of trials per second that the network is performing trying to solve the mining puzzle. In order to maintain stability in the blockchain, the mining puzzle difficulty is dynamically adjusted so that, on average, the time between the creation of two consecutive blocks is constant\footnote{On the Bitcoin blockchain for instance, the block interval is set to 10 minutes and the difficulty adjusts every 2016 blocks which is, on average, 2 weeks. Note that those adjustments are not made by \enquote{someone} but are part of the consensus rules. In other words, the difficulty changes for all miners at the exact same time.}. Therefore as the hashrate increases, the difficulty increases so that it is required to compute more hashes for a given block. Figure \ref{hashRate} represents the estimated total instantaneous bitcoin hashrate since early 2011\footnote{Data are from the website {\color{blue}\emph{blockchain.com}} but the can be recovered from the difficulty and the actual inter block time, which can be both read within the blockchain.}. 

\begin{figure}[h]
\includegraphics[scale=0.4]{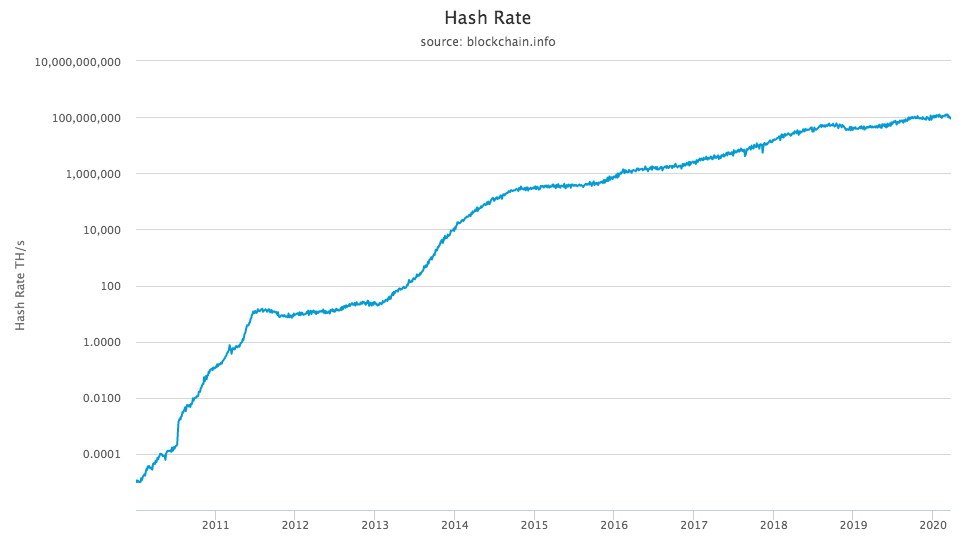}
\caption{Bitcoin Hashrate (tera hashes per second - log scale)}
\label{hashRate}
\end{figure}

Our goal is to analyze the equilibrium behavior of miners and therefore deduce the equilibrium outcome regarding the hashrate. The total hashrate represents the total computational power devoted to block creation. It is a measure of how harsh the competition among miners is, and is directly linked to the security of the underlying blockchain. Indeed, the most common attack scenario on a PoW chain is the so-called 51\% attack in which an attacker who owns more than 51\% of the total hashrate can make a double spend\footnote{Double spending refers to situation in which an agent is able to spend the same amount of money twice. An attacker with more hashrate than the rest of the network could in principle rewrite part of the chain and perform a double spending. For more information on 51\% attack see : \citeasnoun{grunspan2018double}.}. Therefore the higher the hashrate, the more expensive it is to own more than 51\% of it, and the more secure the blockchain.

The contribution of this paper is twofold. First, we show that the game arising from the PoW algorithm always lead to a unique equilibrium, it is well posed. Second, we analyze the long-run equilibrium, as well as the short-run dynamics, in a wide variety of settings meant to represent different sets of real world characteristics, that could all be combined into a quite general model. The main results in the long run is that the PoW algorithm contains some stability that translates to security, namely that security in the long run is either $i)$ constant in deterministic settings or $ii)$ positively driven by the demand of the underlying cryptocurrency. 

Our model features a continuum of risk neutral miners interacting only through the aggregate computational power they allocate to mining the blockchain. Such a setup can be analyzed within a MFG framework.

The most general way to describe the PoW mining game is as follows. The blockchain produces a deterministic output (expressed and paid in units of cryptocurrency) that miners are competing to get. Miners control the hashrate they provide, which increases their share of the blockchain reward, all other things being equal, while reducing all other miners share of the reward. On the other hand, mining chips consume a lot of energy and need to be supplied with electricity. All miners have access to the same technology for computing hashes. We assume a constant rate of technological progress on the mining hardware, that is the mining chips become exponentially more efficient at converting electricity into hashes.

A key feature of our analysis is the use of the Mean Field Game theory, and in particular the master equation formulation. This allows us to consider equilibrium with a continuum of agents. In general, it is usually possible to solve for MFG equilibrium using a forward-backward system with a Hamilton-Jacobi equation and a Fokker-Planck equation for the evolution of the state variables. In the master equation approach, the dependence to the initial condition is hidden as it considers the value function on the whole domain, it is therefore more general.

To make the analysis as clear as possible, our strategy is to introduce a very simple version of the mining game, with a set of restricting assumptions, and show that the mathematical formulation of the model is well-posed and intuitive. Then we present a set of extensions that represent more accurately real-world conditions, for which the structure of the model is preserved as well as the main properties. With this approach we hope to convince the reader that our formulation of the model is general enough and that it can be tuned to fit different sets of assumptions.

There are three assumptions that are common across all extensions. First, we assume that each computed hash yields a share of the total reward. In practice, the blockchain reward is discrete as it is paid to the wining miner every time a block is solved. As a consequence, a miner contributing a small portion of the total hashrate will likely have to wait for a while before getting his share of the reward. This delay might cause him to run out of money because ke will likely have to sell the cryptocurrencies earned as reward to pay for the mining cost (electricity). To hedge this risk, miners have the possibility to pool their mining resources, and share the reward accordingly. We abstract from frictions that can happen on the pooling market. Second, we assume that the mining difficulty continuously adjusts. There are short term frictions that can arise from the delay in difficulty adjustment. Depending on the protocol, the mining difficulty adjusts more or less frequently, but it is always discrete. This can make the actual inter-block time to diverge temporarily but significantly from its long-term average. Lastly, we assume the block reward is constant over time. In practice, there are two factors that can change the block reward : $i)$ coin issuance schedule, and $ii)$ transaction fees. Take Bitcoin for instance. Within each block, some newly created pieces of cryptocurrency are awarded to the successful miner. Because it was designed to have a maximum supply of cryptocurrency, the coin issuance is halved every 210,000 blocks, which corresponds to roughly 4 years\footnote{The issuance schedule is different for every blockchain but note that there are usually long periods for which it is constant.}. Also, transaction fees can change a lot over time depending on the congestion in the pending pool of transactions. In practice, as of early 2020, transaction fees still have a low proportion in the total block reward, and except around the halving period, the block reward can be approximate by a constant.

In the majority of the paper we assume that miners continuously buy and sell mining hardware but are constrained on how fast they can change the hashrate they provide. On the buying side, this constraint represents a direct cost of increasing the mining capacity, such as buying a new warehouse to store the mining rigs, or the need to invest in security for large mining facilities. These constraints could also come from the mining chips manufacturers side, that is the production of mining chips is not infinite and there might be some shortage when the demand is high. No matter were this constraint comes from, it makes the total hashrate to adjust continuously to profit opportunity at a finite elasticity. For the sake of clarity we also make the same assumption on the selling side, that this the mining rigs can be sold if the expected profit is negative but with some frictions that prevent jumps in the hashrate. We provide an extension of the baseline model in which we show that our results are robust to different interactions on the market for mining hardware.

Regarding the hashrate, because of the assumed technological progress, it turns out it is more convenient to focus on the hashrate discounted by the rate of technological progress, which we call the \emph{real hashrate} by opposition to the true hashrate which we call the \emph{nominal hashrate}. Indeed, as the ability to convert electricity into hashes increases, it becomes increasingly cheaper to reproduce the same nominal hashrate, while a constant real hashrate represents a constant cost for miners, everything else being equal.

An important component of the PoW algorithm is that the currency of the reward is different from the currency used to pay for the cost. Indeed, the blockchain pays out a reward in cryptocurrency (reward currency) while electricity and mining hardware needs to be paid with traditional fiat currency like the Dollar or the Euro (cost currency). Therefore the price of the reward currency relative to the cost currency is a major determinant of the incentive to mine\footnote{As it will become clearer in the model, the cryptocurrency price is the main mechanism by which the incentives to mine will change over time. Since the supply of cryptocurrency is fully deterministic, the demand on the secondary market will set a market price, which will in turn give miners more or less incentives to mine. This is the core of the PoW consensus algorithm.}. 

In the most simple version of the model, we take the price of the underlying cryptocurrency as constant. As an intuition, this could also corresponds to a potential long term situation in which the cryptocurrency market is mature and the volatility of the cryptocurrency has dropped to close to 0, or similarly a situation in which electricity can be paid in cryptocurrency. This assumption is relaxed in on of the extensions, but it will make the first analysis clearer. A continuum of risk neutral miners continuously adjust their hashrate. As explained above, investment decisions are such that the hashrate is continuous, and the mining hardware cannot be turned off. This simplistic structure yields a clear analysis of the equilibrium. Note that when the cryptocurrency price is constant the model is deterministic.

Our first result is to show that the mining game arising from the PoW algorithm, as previously described is well posed. We show that the value function of one unit of real hashrate always exists and is unique. Note that this model is stationary and the only state variable is the total hashrate expressed in real terms. Furthermore, to characterize the dynamics of this real hashrate, we show that, regardless of the initial condition, it eventually reaches a stationary state. At this level, miners collectively add new machines, just to compensate for the depreciation due to the technological progress. In other terms, the real hashrate is constant, and the nominal hashrate increases at the rate of the technological progress. In the most simple version of the model, we obtain a closed form solution for the value of the stationary state, and we provide some comparative statics on $i)$ the value of the hashrate, $ii)$ the value generated by each unit of hashrate, and $iii)$ the total value generated by all the miners. In particular, we show that as the level of friction in the market for mining hardware decreases, the number of mining machines (the hashrate) increases. On the other hand the value generated by each machine decreases to reach zero in the limit case with no friction. Therefore the total value generated by miners decreases with the level of friction. Also, when the rate of technological progress increases, it decreases the hashrate while increasing the value generated by each mining machine. The overall effect on the total value generated by all miners is ambiguous. For low (high) value of technological progress, the total value generated by miners increases (decreases) with the technological progress. There seems to be a rate of technological progress that maximizes the total value generated by miners. In other words, from the point of view of miners, some technological progress is desirable but not too much.

We then provide a set of extensions to show how the master equation approach can help us refine the model to fit a more realistic set of assumptions. In general, the structure, as well as the regularity of the model is preserved. As a first extension, we introduce some randomness in the exchange rate between the reward currency (the cryptocurrency) and the cost currency. We show that the equilibrium is still unique and well defined. With a random price, we prove existence of a function that gives a target value for the hashrate for each price. This is the equivalent of the stationary state in the deterministic setting. The randomness will keep the target hashrate moving around (smoothly), and the game forces will push the real hashrate toward this moving target. Second, we present an extension with two different populations in the deterministic setting. In particular, the two populations face two different prices for electricity. We show that the best way to analyze such a setting is to consider the master equation with two state variables (the real hashrate for each population). We prove existence and uniqueness of the equilibrium. Like in the setting with a single population, the real hashrate will eventually converge toward a stationary state. Depending on the difference in electricity price, the stationary state can be located on the boundary of the domain. In other words, if the difference in electricity prices is large enough, only miners with access to the cheap source of electricity will actively mine the blockchain. Third, we present a two-population setting with a random exchange rate between the currency used by the two populations. We show how to consider such a setting and extend the results of the previous extensions. For every value of the exchange rate, there is a two-dimensional stationary state that will attract the real hashrate. Lastly, we study the case when there is no more friction on the computing machines market, but the machines have a zero resale value. In this case, because of free entry, the value function cannot be positive. We show that there is still a unique equilibrium. The real hashrate still reaches a stationary state which corresponds to the limiting case when the level of friction tend to zero. Because of free entry, if the initial condition is such that the hashrate is below the stationary state, it will instantly reach it. However, because the machines cannot be sold, if the initial hashrate is higher than the stationary state, it will decrease at the rate of technological progress until it reaches the stationary state.

This paper builds on the growing blockchain literature. Since \citeasnoun{nakamoto2008bitcoin}, blockchains have attracted more and more interest from academics. Some papers, like \citeasnoun{biais2018equilibrium}, \citeasnoun{schilling2019some} and \citeasnoun{pagnotta2018equilibrium} study the underlying economics of bitcoin and analyze the implied valuation for cryptocurrencies. We rather focus on the economics of the consensus algorithm and take the price of the cryptocurrency as exogenous.

There is a growing literature studying frictions that can arise on a blockchain. \citeasnoun{easley2019mining} build a model for analysing the emergence of transaction fees, and \citeasnoun{huberman2017monopoly} study the cost structure of a blockchain based transactional system. Transaction fees do not play a great role in our analysis and we take them as constant. As of early 2020, transaction fees are a small proportion of the block reward. In the future, as the supply of new coins will decrease, it is possible that transaction fees represent a higher share of the block reward, but our model will work the same as long as the fees volatility is low.

Some authors have studied optimal contracting under mining pools in which miners gather to mitigate the risk of not finding a block (see for instance \citeasnoun{schrijvers2016incentive}, \citeasnoun{fisch2017socially}, and \citeasnoun{cong2019decentralized}). In line with the results of the mining pool literature but for reasons related to R\&D investment, \citeasnoun{alsabah2019pitfalls} show that the mining industry will tend to be centralized. While those questions are an important part of blockchain economics, we take a more general approach and focus on the fundamental mechanisms. Indeed, we assume miners are perfectly diversified (or equivalently focus on the long-run) such that a share of $x\%$ of the total hashrate yields $x\%$ of the reward. This is as if every single hash computed was remunerated to its marginal reward. \citeasnoun{biais2019blockchain} study the game theoretic aspects of forks, showing that there are multiple Nash equilibria in which forks may exist. We analyze the mining game from the computational intensity point of view, and in our model there is a single chain miners can mine.

This paper uses a Mean Field Game framework to model the mining game. Since its introduction by \citeasnoun{lasry2007mean}, the MFG theory has been quite extensively developed. In particular, there are some recent works on MFG involving different populations like \citeasnoun{cirant2015multi}, \citeasnoun{achdou2017mean}. We uses similar results to analyze the mining game when two different populations of computing hardware co-exist. \citeasnoun{li2019mean} also analyze the mining game in an  MFG setup however their approach is quite different. Indeed, they focus on the risk borne by risk-averse miners and study mining concentration. Our paper assumes all miners are risk-neutral and perfectly diversified and focuses on the total hashrate stability. Also, instead of using the classical forward-backward system, we study the so-called master equation, which is probably the most general way of studying such games (see \citeasnoun{cardaliaguet2019master}).

With an analysis closely related to ours, \citeasnoun{prat2018equilibrium} build a structural model to estimate the underlying unobserved parameters. Like us, they take a general approach and take miners as diversified risk neutral agents, and the share of the reward they get is strictly proportional to the hashrate they contribute. Their main assumption is free entry in the mining industry when there is a profit opportunity. Because of variations in the price, buying new mining machines is not always profitable but when it is profitable, they show that in equilibrium a number of miners enter the game so that the expected profit is reflected along the zero profit line. This produces an equilibrium hashrate that is piece-wise constant. Our approach could include such a case by combining two of the extensions we provide : the random price extension and the last extension that assumes free entry and no friction in the hardware market. With our approach, uniqueness for their model can be proven.

The remaining of the paper is organized as follows. Section \ref{mod} presents the basic formulation of the model, while section \ref{mainRes} analyzes the resulting master equation and the mining dynamics. Section \ref{exts} provides a set of extensions of the base model. Section \ref{masterVsHJB} comments on the use of the master equation. Section \ref{conclusion} concludes.

\section{The model}
\label{mod}

This section describes the most simple version of the model. We focus on the core mechanisms of the PoW consensus algorithm. Most of the simplifying assumptions will be relaxed in section \ref{exts}.

Let's define by $P$ the hashrate measured as the number of hashes computed per second, hereafter we use the term \emph{nominal hashrate}, or \emph{hashrate in nominal terms}. Let $\delta$ be the rate of technological progress of the mining machines, that is how much more hashes can new generations of hardware compute for a given cost. In other words, it represents the fact that for the same cost a miner can get more hashes per second, it becomes cheaper to buy the same computational power. To make the analysis clearer we focus our attention on the progress-adjusted hashrate, let's call it the \emph{real hashrate}, and denote this real hashrate at time $t$ by $K_t$.
\[
K_t\coloneqq e^{-\delta t}P_t.
\]

The real hashrate, $K_t$, is the nominal hashrate discounted by the rate of technological progress. Otherwise stated, it is the number of current machines needed to rebuild the entire hashrate. Variations in the real hashrate, $K$, represent variations in the hashrate not related to the technological progress. The following consists of analyzing the dynamic behavior of $K$ in a continuous-time model. 

Let us comment on the importance of the real hashrate. On a PoW chain the hashrate (the true one, i.e. the nominal hashrate) is often associated to the blockchain security. This is related to the most common way to attack a PoW chain that is the so called 51\% attack. If an attacker was able to control at least the majority of the total hashrate he could in principle rewrite the blockchain which is of course against the core principle of consensus\footnote{Note that an attacker with a hashrate slightly higher than 51\% cannot go \enquote{back in time} in the sense that he cannot rewrite old blocks deeply buried into the blockchain. The more hashrate he has, the older the block he can change. For more details, see \citeasnoun{nakamoto2008bitcoin} and \citeasnoun{grunspan2018double}.}. There are two ways for someone to end up with more than 51\% of the hashrate : $i)$ an internal attack in which existing miners would collude to obtain the required majority, and an $ii)$ external attack in which an outside agent would buy enough mining hardware to obtain the majority. Let us focus on the latter as more decentralization among miners should help prevent from internal 51\% attack but it is out of the scope of this paper\footnote{For a discussion on collusion attacks see \citeasnoun{eyal2014majority}.}. The important aspect of a 51\% attack is its cost. Because of the rate of technological progress, if the nominal hashrate were to be constant, as time goes on it would become cheaper and cheaper to buy and run hashing devices. It is therefore more convenient to use the real hashrate as it has a direct connection to the cost of a 51\% attack and therefore to the security of the blockchain. Indeed, recall that the real hashrate represents the number of current machines needed to re build the entire nominal hashrate. We address this question of security in section \ref{comparativeStatics}.

Let's consider a model in which the machines cannot do anything but mine the blockchain, and the price of the associated cryptocurrency is constant and normalized to 1. As discussed in the introduction, this is the most simple version of the model and those assumptions will be relaxed in section \ref{exts}.

We assume that miners discount time at rate $r$. Miners continuously buy new computation units to compute hashes\footnote{On the most popular blockchains miners use Application Specific Integrated Circuits (ASICs) to compute hashes. The hashing algorithm is hard-coded into the machines and a lot of speed efficiency can be obtained. For the sake of simplicity we refer to hash-computing machines as \enquote{mining machines}, \enquote{mining devices}, \enquote{computing chips} or even \enquote{machines}.}. The mining game works as follows. There is a fixed number of coins output by the blockchain per unit of time, miners compete against each other to earn a share of this fix output. The share they get is proportional to their relative share of the total computational power.

Let $U$ be the value of 1 unit of real hashrate. We assume that there are no barriers to entry in the mining market, however there are some frictions in the mining devices market affecting indistinctly all miners (current miner and potential new miners). Whenever there is a profit opportunity associated to mining the blockchain, i.e. $U>0$, miners cannot simply buy an infinite amount of computing machines. This friction could be justified through capacity constraints on the mining device manufacturers side. Indeed, mining equipments are frequently sold out and so even if miners want to buy more machines they might need to wait for the manufacturer to produce more machines. Note that in this version of the model, we assume that the effect is exactly similar when $U<0$, i.e. miners are reselling their hardware but with some constraints. Having a symmetric constraint helps understanding the underlying mechanisms, we show in section \ref{2p} that this symmetric assumption can be relaxed.

Let $\lambda$ be the parameter that represents this investment constraint, that is the continuous flow of entry of computing devices is given by $\lambda U$. In addition to that, because of the technological progress, the real hashrate depreciates at rate $\delta$. This represents the fact that if nothing happens it becomes cheaper to buy the same nominal hashrate and therefore less secure. This produces the following dynamic for the real hashrate $K_t$.

\begin{equation}
\dot{K}_t=-\delta K_t + \lambda U_t.
\label{dynK}
\end{equation}

Because $\dot{K}$ depends on $U$, we need to understand the value of the real hashrate. On the one hand, a unit of real hashrate yields a revenue of $\nicefrac{1}{K}$. Indeed, recall that the blockchain pays out a fix amount per unit of time (here normalized to 1), and the revenue of a miner is proportional to its share of the computing power (the nominal hashrate). Note that we assume perfect diversification of miners. Because of the structure of the blockchain (discrete payouts at each block creation), the actual revenue of any cluster of hashrate will be awarded in big lump sum whenever a block is found, which follows a poisson distribution. Also the difficulty of the mining puzzle self-adjusts in a discrete fashion as well. However, we argue that those frictions only have an impact in the short run and can be diversified away\footnote{See for instance the use of mining pools, \citeasnoun{schrijvers2016incentive}, \citeasnoun{fisch2017socially}, and \citeasnoun{cong2019decentralized}.}. We intend to give a general formulation of the PoW algorithm therefore we assume that one unit of real hashrate yields a revenue of the form $\nicefrac{1}{K}$. 

On the other hand, in order to run the machines miners need to pay for electricity, of which the cost is denoted by $c$. Obviously the revenue diverges when $K$ goes to $0$. To avoid such technicality in the rest of the paper, we assume that one unit of real hashrate yields a revenue of $\nicefrac{1}{K+ \epsilon}$ where $\epsilon > 0$ is a small parameter ($\epsilon << c$). Let us note that $\epsilon$ can be interpreted as the size in real hashrate of a single machine. Note that in this section we assume the revenue and cost of mining are of the same currency. In practice the cost of mining (electricity) is paid in traditional national currency and the reward is in the blockchain's cryptocurrency. This has important implications regarding incentives of miners to behave properly. Section \ref{randomPrice} presents a version of the model with a random exchange rate.

Therefore we deduce that the value function $U$ of a unit of real hashrate is given by

\begin{equation}\label{defU}
U(K) :=  \int_0^{\infty}e^{-(r+ \delta)t} \left(\frac{1}{K_t + \epsilon} - c\right)dt,
\end{equation}

where $(K_t)_{0\leq t \leq \tau^{\delta}}$ is the process which satisfies
\begin{equation}\label{edo}
\begin{cases}
dK_t = -\delta K_t dt + \lambda U(K_t) dt,\\
K_0 = K.
\end{cases}
\end{equation}

If $U$ is smooth, it clearly satisfies 
\begin{equation}
0=-(r+\delta)U+(-\delta K+\lambda U)U'_K+\frac{1}{K+ \epsilon}-c \text{ in } [0,\infty).
\label{masterEq}
\end{equation}

This PDE is known as a "master equation" in the MFG literature, more specifically a master equation in finite state space. The next section deals with the study of such a PDE.

This master equation summarizes all the mechanics of the mining game, from the point of view of one unit of real hashrate. Instead of writing down the agents maximization program (which is hidden inside the term $\lambda U$), we can focus on individual unit of computing power and derive this standalone equation.

\section{Analysis of the master equation and the mining dynamics}
\label{mainRes}

\subsection{Equilibrium characterization}

The study of master equations in finite state space is now a well referenced topic in the so-called monotone setting. However the stationary setting has not been particularly studied, neither the fact that this equation is posed only on the half line and not on the whole real line. All the ideas used in the following proofs are already known in the MFG literature, but we still provide them to highlight some structure of the MFG and to precise that our model is indeed well posed. For more details on master equations in finite state space, we refer to \citeasnoun{lions2007cours} and \citeasnoun{bertucci2019some}.

Let us recall that the general structure known to yield well-posedness for MFG is the monotone structure. The mathematical analysis we provide here relies heavily on this monotone structure. Our main result concerning (\ref{masterEq}) is the following.

\begin{theorem}
\label{thm1}
There exists a unique lipschitz function solution of (\ref{masterEq}). This function is decreasing and satisfies $U(0) \geq 0$ and $\lim_{K \to \infty}U(K) \leq 0$.
\end{theorem}

\begin{remark}
Let us note that no boundary condition is neither imposed nor needed at $\{ K = 0 \}$. This is due to the fact that the set $\{K \geq 0\}$ is invariant under the dynamic of the problem.
\end{remark}

\begin{proof}

We begin with the uniqueness and decreasing properties. Let us assume that there exists $U$ and $V$ lipschitz functions solutions of (\ref{masterEq}) which satisfy the prescribed behavior at $0$ and infinity. Let us define the function $W : [0,\infty)^2 \to \mathbb{R}$ by 
\begin{equation}
W(x,y) = (U(x) - V(y))(x-y).
\end{equation}

The function $W$ is a solution of 
\begin{equation}
\begin{aligned}
(r + \delta) W = & (- \delta x + \lambda U(x))\partial_x W + (- \delta y + \lambda V(y))\partial_y W\\ 
& + \left(\frac{1}{x + \epsilon} - \frac{1}{y + \epsilon}\right)(x-y) + \delta W - \lambda (U(x) - V(y))^2.
\end{aligned}
\end{equation}

A simple maximum principle type result, that we do not detail here, yields that for all $x, y \geq 0$ 
\begin{equation}
W(x,y) \leq 0.
\end{equation}

This implies first that $U = V$. Indeed, if there is $\bar{K}$ such that $U(\bar{K})\ne V(\bar{K})$, then the sign of $W$ should change around $(\bar{K},\bar{K})$, which is not the case. Therefore $U= V$, from which we deduce that $U$ is decreasing.

Concerning the existence and boundary behavior of the solution of the equation, we are not going to provide a complete proof of this facts, as they are classical in the PDE literature. We only prove the main argument of the proof, which is a lipschitz a priori estimate for the solution of (\ref{masterEq}). Let $U$ be a smooth solution of (\ref{masterEq}), denote its derivative $V= \partial_K U$ and let us consider the function $W: [0,\infty) \to \mathbb{R}$ by
\begin{equation}
W(K) = V(K) + \beta V^2(K),
\end{equation}

for some $\beta > 0$ to be defined later on. Our aim is to show that there exists $\beta$ such that $W\leq 0$ on $[0,\infty)$. Let us remark that $W$ satisfies
\begin{equation}
\begin{aligned}
(r + \delta)W = &(-\delta K + \lambda U) \partial_K W  - \frac{1}{(K + \epsilon)^2} + 2 (-\delta + \lambda V)W - V(-\delta + \lambda V)\\
& - (r + \delta )\beta V^2 - 2 \beta V \frac{1}{(K + \epsilon)^2}.
\end{aligned}
\end{equation}

Let us now take $\beta$ such that 
\begin{equation}
\beta \leq \frac{\epsilon^2\delta}{2}.
\end{equation}

We then obtain that
\begin{equation}
\begin{aligned}
(r + \delta)W -  \partial_K W (-\delta K + \lambda U) -  2 (-\delta + \lambda V)W \leq- \frac{1}{(K + \epsilon)^2} \hspace{1cm}\\
\hfill
-  \lambda V^2 - (r + \delta )\beta V^2.
\end{aligned}
\end{equation}

Once again, a maximum principle result type yields that $W \leq 0$ and thus that
\begin{equation}
-\beta^{-1} \leq V \leq 0,
\end{equation}

which is the announced lipschitz a priori estimate.
\end{proof}

\begin{remark}
The monotone structure of the model is here the fact that the reward $K \to 1/(K + \epsilon)$ is decreasing in $K$ and the function which yields the derivative of the hashrate is an increasing function of $U$, i.e. $U \to -\delta K + \lambda U$ is increasing in $U$. 
\end{remark}

The existence and uniqueness of the function $U$, which is the value function of one unit of real hashrate, translate to the same properties for the game equilibrium. This proves that in this context proof-of-work is well-posed. A strength of our approach is that similar kind of arguments can be made to extend this existence and uniqueness results to variation of the master equation (\ref{masterEq}). This will be presented in section \ref{exts}.

Uniqueness of equilibrium being shown, we now explore the properties of the equilibrium real hashrate, which is a crucial for the stability of a blockchain. From (\ref{dynK}) and the fact that $U$ is decreasing, we deduce that as $K_t$ increases, $\dot{K_t}$ decreases. To further investigate the properties of $\dot{K}$ we are interested in the existence of a stationary state. The following proposition proves that there always exists one.

\begin{proposition}
\label{statState}
When the master equation is defined as in (\ref {masterEq}) and the dynamic of the real hashrate, $(K_t)_{t\geq0}$, follows (\ref{dynK}), there always exists a stationary state $K_*$, which can be computed explicitly in terms of the parameters of the model. Moreover, for any $K_0$,
\[
\lim_{t \to \infty} K_t = K_*
\]
\end{proposition}

\begin{proof}
\label{proofStationaryState}
The existence and uniqueness of the stationary state is an immediate consequence of the facts : $U$ is decreasing and $U(0) \geq 0$. Indeed stationary states are characterized by 
\[
\dot{K}(K_*)=0\hs\hs\iff\hs\hs\delta K_*=\lambda U(K_*).
\]

We now compute the value of $K_*$. By the definition of the master equation (\ref{masterEq}) we have that
\[
U(K_*)=\left(\frac{1}{K_* + \epsilon}-c\right)(r+\delta)^{-1}.
\]

Therefore, by combining the two we can write
\[
\delta K_*=\frac{\lambda}{r+\delta}\left(\frac{1}{K_* + \epsilon}-c\right).
\]

Which can be rewritten as
\[
0=\delta K_*^2+K_*\left(\delta \epsilon +\frac{c\lambda}{r+\delta}\right)+  \frac{\lambda}{r+\delta}(c\epsilon - 1).
\]

This second order polynomial equation has always a unique strictly positive solution (since $\epsilon << c$) given by
\begin{equation}
\label{equstat}
K_*=\frac{\sqrt{(\delta \epsilon - \frac{c \lambda}{r + \delta})^2 + 4 \frac{\delta \lambda}{r + \delta} } - \delta \epsilon - \frac{c\lambda}{r + \delta}}{2 \delta}>0.
\end{equation}

Concerning the convergence of the induced trajectories, let us state that it is an immediate consequence of the existence of a stationary state and the decreasing property of $U$. 
\end{proof}

\begin{remark}
Let us comment on the fact that the main idea of this proof is to make an extensive use of some monotonicity properties of the situation we are modeling, namely that $1/(K + \epsilon)$ is a decreasing function of $K$ and $\partial_t K$ is an increasing function of $U$.
\end{remark}

The structure of the mining game produces a stationary state at which new units of real hashrate continuously enter the mining game but only to compensate for the technological progress. Therefore the real hashrate is constant, and the nominal hashrate, $P$, increases exponentially due to the technological progress. More generally, in equilibrium the nominal hashrate follows the rate of technological progress.

We also show that the dynamics of the real hashrate will converge toward the stationary state $K_*$, regardless of where it starts. If the real hashrate is higher than $K_*$, there are \enquote{too many} machines, so the value of one unit of real hashrate decreases (because one unit becomes less important) so miners have less incentives to buy new machines, therefore the real hashrate decreases. If the real hashrate is lower than $K_*$, there are \enquote{too few} machines, so the value of one unit of real hashrate increases (because one unit becomes more important), so miners have more incentives to buy new machines, therefore the real hashrate increases.

The existence and convergence results of a stationary state for the real hashrate mean that the security of the blockchain is constant over the long run. Recall that the real hashrate represents the total number of machines which is directly linked to the cost of reconstructing the whole nominal hashrate with the current technology. Because at the stationary state, the real hashrate is constant, the cost is also constant. The nominal hashrate, which is the observed one, will increase but just to compensate for the depreciation due to the technological progress. The security against external 51\% attack is constant.

\subsection{Analysis of the stationary state}
\label{comparativeStatics}

Now that convergence toward a stationary state is established, we perform some comparative statics to further investigate its properties. 

As explained above, the real hashrate represents the number of current computing devices needed to rebuild the entire nominal hashrate, it then has a direct link to the cost of a 51\% attack and therefore to the security of the blockchain\footnote{We refer only to the cost of buying the hardware needed to perform a 51\% attack. Keep in mind that the machines will then have to be supplied with electricity. Moreover, as explained in the introduction, there are two types of 51\% attack : $i)$ an internal attack (or collusion attack) and $ii)$ an external attack. Here we only refer to external 51\% attack.}. Recall that the price of one unit of hashrate is not taken into account here, which makes the security analysis indirect. Everything else being equal, a higher hashrate requires more ressources to be reproduced.

From the traditional industrial organization literature, the price is likely subject to a minimum value that corresponds to the marginal production cost (in case the market is perfectly competitive) plus some rent that may arise if the market is not perfectly competitive. In the case of a monopoly, the sole manufacturer will likely extract all the value from the miners and therefore price the unit of hashrate at $U_*$. From the previous result, we know that in equilibrium the value of one unit of real hashrate is at a stationary state. From (\ref{dynK}) we know this is equivalent to
\begin{equation}
U_*=\frac{\delta K_*}{\lambda}.
\end{equation}

Therefore we can write the total value generated by all miners in equilibrium, that can be defined by $\Pi_*=K_*U_*$. This corresponds to the maximum rent that could be extracted by manufacturers, and it corresponds to miners making no profit. To assess whether this rent will be kept by miners or extracted by computing chips manufacturers, we would need to precisely model the manufacturing market, which we leave to future research.

The stationary real hashrate can be analyzed as a function of the model parameters with (\ref{equstat}), while the total value generated by miners can be expressed as
\begin{equation}
\Pi_*(r, c, \delta, \lambda)=\frac{\delta K_*^2(r, c, \delta, \lambda)}{\lambda}.
\end{equation}

Figure \ref{comparativeLambda} presents the comparative statics of the quantities of interest with respect to the level of friction $\lambda$, while figure \ref{comparativeDelta} deals with variations in the rate of technological progress $\delta$. All this analysis is made all else being equal. For the calibration of the parameters we choose an energy cost of $c=\$0.02$ per kWh, a discount rate of $r=5\%$, a rate of technological progress of $\delta=20\%$, and a level of friction of $\lambda=1$. We also assume $\epsilon=0$. This calibration is not meant to produce quantitative results but rather to qualitatively present the intuition of the model.

\begin{figure}[h]
\center
\includegraphics[scale=0.5]{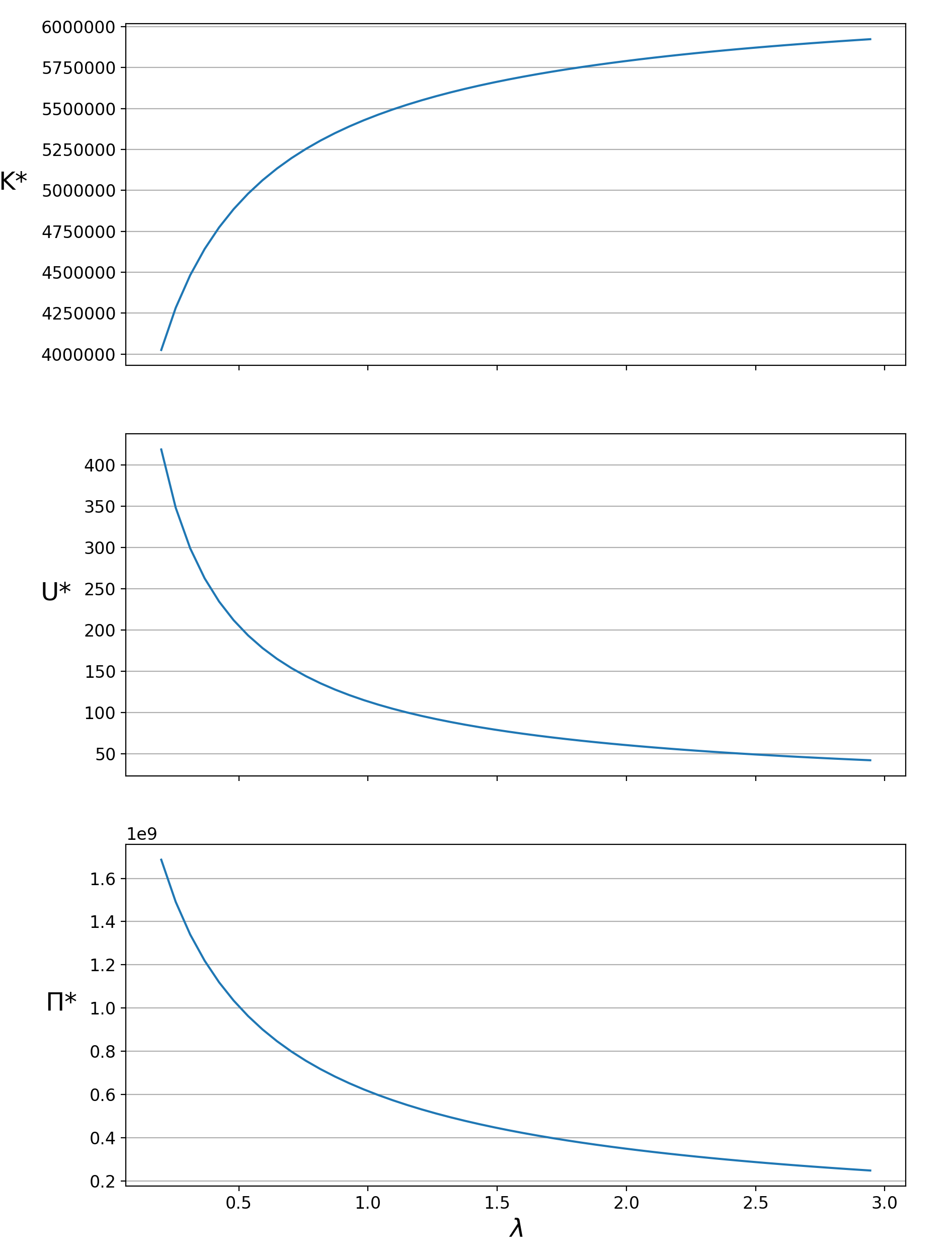}
\caption{Comparative statics of equilibrium quantities with respect to the level of friction $\lambda$ - more precisely $\lambda$ is the inverse of market frictions}
\label{comparativeLambda}
\end{figure}

\begin{figure}[h]
\center
\includegraphics[scale=0.5]{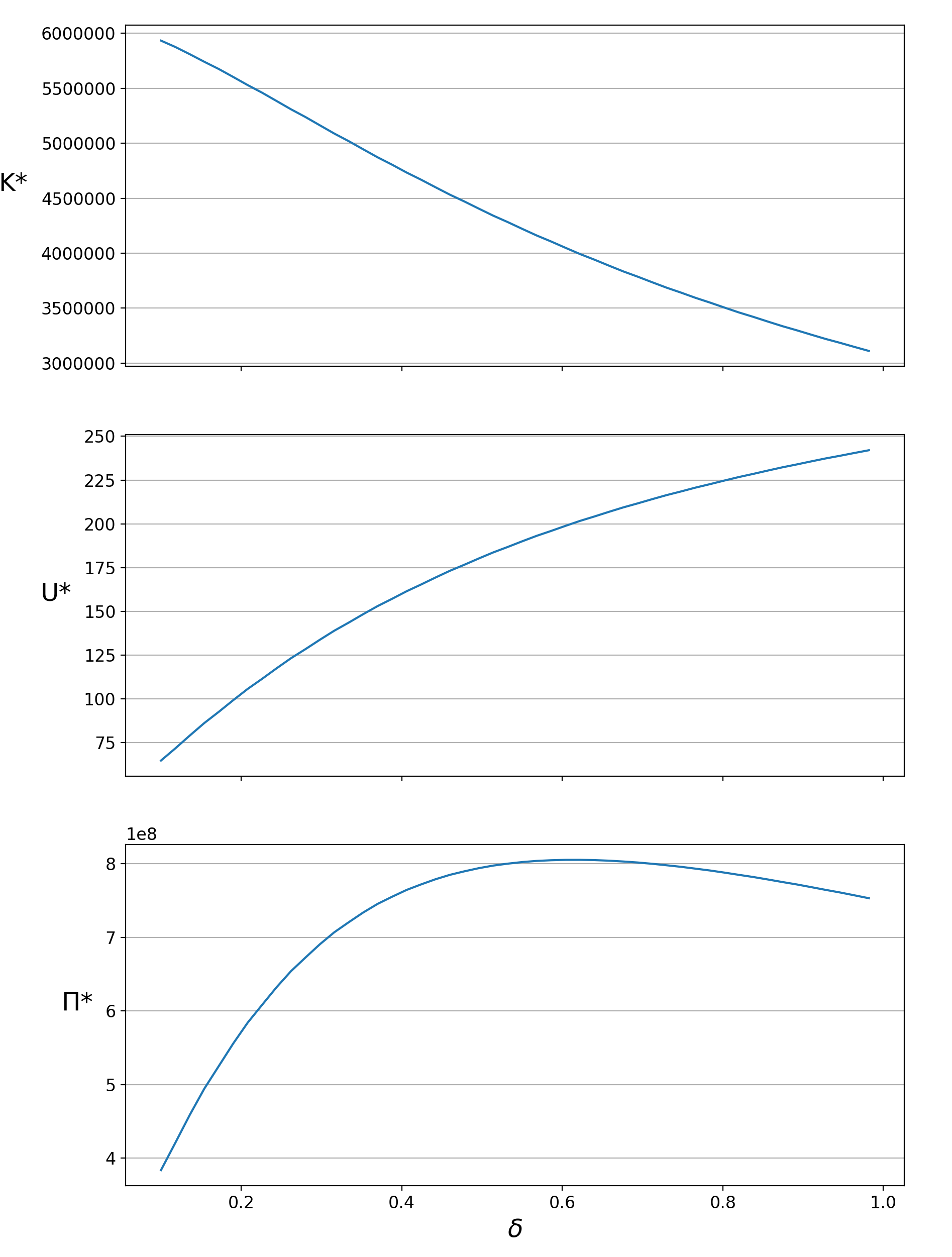}
\caption{Comparative statics of equilibrium quantities with respect to the technological progress $\delta$}
\label{comparativeDelta}
\end{figure}

In figure \ref{comparativeLambda} we see that as the level of friction decreases ($\lambda$ increases), the total value generated by the miners ($\Pi_*$) decreases. There are two contradictory effects. On the one hand, a higher $\lambda$ implies a higher stationary state for the real hashrate. When it is easier for miners to buy and sell hashing machines, they buy more in equilibrium. This means that the security of the blockchain increases. On the other hand, miners are able to enter and exit the market with less constraints, and they also anticipate future miners to easily join the industry if there is still a profit opportunity, and as a result the price of one unit of real hashrate ($U_*)$ decreases. In section \ref{2p}, we analyze an extension which corresponds to the limit case of $\lambda\rightarrow\infty$, in which we show that at the stationary state $U_*=0$ and the real hashrate tends toward $K_*=\frac{1}{c}-\epsilon$.

Let's now turn to the analysis of variations in the rate of technological progress $\delta$. In figure \ref{comparativeDelta}, with the parameters we calibrated, we see that the total value generated by miners seems to reach a maximum at around $\delta=60\%$. We see in the upper panel that the real hashrate, and hence the security of the blockchain, decreases as $\delta$ increases. This is because the rate of technological progress makes recent machines more performant than older machines for the same cost. Because machines will depreciate more, miners have less incentives to buy them in the first place.

The effect on the value of one unit of real hashrate is reversed. There are two effects at play here. First, a higher technological progress induces a lower real hashrate which makes the value of the hashrate to increase, because it is easier to mine. Second, the units of real hashrate depreciate more which will drive the value down. Clearly in the middle panel of figure \ref{comparativeDelta}, we see that the first effect dominates for the value of the hashrate $U_*$. However, when we assess the total value generated by miners (in the lower panel) the effect is not clear. For low values of technological progress, the effect on the hashrate dominates and the total value increases with $\delta$. However, at some point the rate of technological progress depreciates too much the current machines and this effect starts to dominate. For high values of $\delta$ the total value generated by the miners decreases with the technological progress. From the point of view of the miners, some level of technological progress is desirable but not too much.

\section{Natural extensions of the model}
\label{exts}

In this section we describe some natural extensions of our model. All theses extensions could be mixed together and we present most of them separately for the sake of clarity. An important objective of this section is to highlight the fact that the monotone structure of the MFG plays a key role in the well-posedness of these models, and that our approach allows for a great deal of flexibility.

\subsection{Random exchange rate between the reward and the cost}
\label{randomPrice}

In the version of the model we presented in section \ref{mod}, the cost paid by the players and the reward they gain from mining the blockchain were assumed to be in the same currency, or equivalently, in currencies with a constant exchange rate. In this section we assume that this exchange rate is random and given exogenously. More precisely we assume that the reward $1/(K_t+\epsilon)$ at time $t$ is replaced by

\begin{equation}
\frac{g(P_t)}{K_t+\epsilon},
\end{equation}

where $g$ is a given positive smooth function and $P_t$ evolves according to 

\begin{equation}
\label{dynP}
dP_t = \alpha(P_t)dt + \sqrt{2\nu}dW_t,
\end{equation}

where $\alpha$ is a lipschitz function, $\nu > 0$ and $(W_t)_{t\geq 0}$ is a standard brownian motion on a probability space $(\Omega, \mathcal{A}, \mathbb{P})$. In this context, $(P_t)_{t\geq 0}$ is a process which captures the randomness of this exchange rate. Note that with this formulation, the value of the reward in the national currency $g(P_t)$ can follow any type of random process. In particular, if $\alpha$ is constant and $g$ is the exponential function then the process $\left(g(P_t)\right)_{t\geq0}$ is a geometric brownian motion.

This type of models falls under the terminology of MFG with common noise. In such a context, the value function $U$ of a unit of real hashrate depends on the value of the process $(P_t)_{t\geq 0}$. The analogous of (\ref{defU}) and (\ref{edo}) becomes here

\begin{equation}
U(K,p) := \mathbb{E}\left[ \int_0^{\infty} e^{-(r+\delta)t} \left(  \frac{g(P_t)}{K_t+\epsilon} - c \right)  \right],
\end{equation}

where $(K_t)_{t \geq 0}$ and $(P_t)_{t \geq 0}$ evolves according to 

\begin{equation}\label{edo2}
\begin{cases}
dP_t = \alpha(P_t)dt + \sqrt{2\nu}dW_t,\\
dK_t = - \delta K_t dt + \lambda U(K_t,P_t)dt, \\
K_0 = K, P_0 = p.
\end{cases}
\end{equation}

In this case, the master equation satisfied by $U$ is

\begin{equation}\label{meqp}
- (r + \delta) U + (-\delta K + \lambda U)\partial_K U + \alpha \partial_p U + \nu \partial_{pp} U + \frac{g(p)}{K + \epsilon} - c= 0 \text{ in } [0,\infty)\times \mathbb{R}.
\end{equation}

In this case of a non constant exchange rate, the essential structure of monotonicity has not been perturbed and the equation (\ref{meqp}) remains well posed. The following result holds true.

\begin{theorem}\label{thmp}
Assume that there exists $A> 0$ such that for all $p \in \mathbb{R}$
\begin{equation}
\epsilon c \leq g(p) \leq A.
\end{equation}
Then, there exists a unique solution $U$ of (\ref{meqp}), it is decreasing and uniformly lipschitz in $K$.
\end{theorem}

\begin{remark}
The bound from below on $g$ models the fact that, whatever the exchange rate, mining is profitable when $K=0$. Such an assumption could have been replaced by the possibility for the miners to turn off the machines or the impossibility to sell them. We refer to the other extensions for more on this kind of assumptions.
\end{remark}

\begin{proof}
As in the previous case, this result relies heavily on monotonicity assumptions. First, let us mention that from the assumption we made on $g$, the domain $[0,\infty)\times \mathbb{R}$ is invariant.

The uniqueness and decreasing properties follow from the maximum principle applied on the PDE satisfied by $W: [0,\infty)^2\times \mathbb{R}$ defined by

\begin{equation}
W(x,y,p) = (U(x,p) - V(y,p))(x-y),
\end{equation}

for $U$ and $V$ two solutions of (\ref{meqp}). 

Concerning the lipschitz property, we introduce the function $W$ defined by

\begin{equation}
W(K,p) = \partial_K U(K,p) + \beta (\partial_K U(K,p))^2;
\end{equation}

for $\beta > 0$ such that 
\begin{equation}
\beta \leq \frac{\epsilon^2 \delta \|g\|_{\infty}}{2},
\end{equation}
the function $W$ satisfies 
\begin{equation}
\begin{aligned}
(r + \delta)W - \partial_K W (-\delta K + \lambda U) - 2 (-\delta + \lambda \partial_K U)W - \alpha \partial_p W - \nu \partial_{pp}W \leq0.
\end{aligned}
\end{equation}

As in the one dimensional case, we conclude by the maximum principle that $W \leq 0$ and thus that $\partial_K U$ is bounded. 
\end{proof}

Unlike the one dimensional case, there is no convergence of the real hashrate to a stationary equilibrium, this is clearly due to the randomness. However, from the monotonicity of $U$ with respect to $K$ we can establish the following :

\begin{proposition}
Under the assumptions of the previous theorem, given the solution $U$ of (\ref{meqp}), there is a smooth function $K_* : \mathbb{R} \to \mathbb{R}$ such that if $(K_t,P_t)_{t \geq 0}$ is given by (\ref{edo2}), then 
\begin{equation}
\begin{cases}
K_t \leq K_*(P_t) \Rightarrow dK_t \geq 0,\\
K_t \geq K_*(P_t) \Rightarrow dK_t \leq 0.
\end{cases}
\end{equation}
\end{proposition}

\begin{proof}
From the properties of  $U$ given by theorem \ref{thmp}, for every value of $p$, $K \to U(K,p)$ is decreasing and $U(0,p) \geq 0$. Thus for every $p$, the equation $\lambda U(K,p) = \delta K$ has a unique solution which we note $K_*(p)$, from the regularity of $U$, we obtain that $K_*$ is a smooth function.
\end{proof}

When the exchange rate between the reward and the cost of mining is random, there is a function that acts like a stationary state for each $p$. In other words, for every value of $p$, there is a threshold $K_*(P_t)$ which characterizes the fact that the real hashrate is increasing or decreasing. The randomness of the reward prevents the real hashrate to stay at a constant level. There are two forces that apply to the hashrate $K$. First, randomness in the reward acts as a series of exogenous random shocks, this can make mining being more or less profitable (depending on the current value of the hashrate $K$), which will push the hashrate either toward or away from $K_*(P_t)$. Second, miners will react to these random shocks by adjusting their computational power toward $K_*(P_t)$. In other words, the function $K_*(P_t)$ acts like an attractor for the hashrate, but the randomness in the reward keeps the equilibrium hashrate moving around.

\subsection{Two populations of miners in a deterministic case}

We present here an extension in the deterministic setting in which there are two types of miners. The case with more than two types of miners can be treated in the same way. We label all the constant of the previous model with an index $i \in \{1; 2\}$ to indicate to which type of miners they refer. The two types of miners have access to the same mining equipment technology, i.e. $\delta_1=\delta_2=\delta$.

This is meant to represent two distinct populations that have access to different prices for electricity, $c_1\neq c_2$. MFG models with several populations is by now a well documented subject and we refer to \citeasnoun{achdou2017mean} and \citeasnoun{cirant2015multi} for more details on this topic. Let us indicate that in such a context, the right approach consists in considering the value function for the two populations as well as the real hashrate produced by the two types of miners (which are the mean field terms here). In this model, we introduce reservation utilities $\alpha_1,\alpha_2 \geq 0$ for the two types of miners. Meaning that they have the possibility to turn off their machines. To account for that we slightly change the notation. We label $K$ the potential real hashrate by the miners of the first type if they are all mining, and $L$ the potential real hashrate by the miners of the second type if they are all mining. In other words, this represents the stock of hashing machines for each population. And we denote by $\phi(K,L)$ and $\psi(K,L)$ the real hashrate actually produced by the two types of miners. The value functions of the two types of miners are denoted by $U$ and $V$ and satisfy
\begin{equation}
\label{defU2}
U(K,L) := \mathbb{E}\left[ \int_0^{\infty} e^{-(r_1+\delta)t} \max \left(  \frac{1}{\phi(K_t,L_t)+\psi(K_t,L_t)+\epsilon} - c_1; \alpha_1 \right)  \right],
\end{equation}

\begin{equation}
\label{defV2}
V(K,L) := \mathbb{E}\left[ \int_0^{\infty} e^{-(r_2+\delta)t} \max\left(  \frac{1}{\phi(K_t,L_t)+\psi(K_t,L_t)+\epsilon} - c_2; \alpha_2 \right)  \right],
\end{equation}

where $(K_t)_{t \geq 0}$ and $(L_t)_{t\geq0}$ evolve according to 

\begin{equation}\label{edo2pop}
\begin{cases}
dK_t = - \delta K_t dt + \lambda_1 U(K_t,L_t)dt,\\
dL_t = - \delta L_t dt + \lambda_2 V(K_t,L_t)dt, \\
K_0 = K , L_0=L.
\end{cases}
\end{equation}

Because we introduced reservation utilities, the miners can decide not to participate in the mining game and get the utility $\alpha_i$. This ensures the positivity of $U$ and $V$ and thus the stability of the domain $\{K \geq 0\}\cap \{L \geq 0\}$ under the trajectories given by (\ref{edo2pop}). For the sake of simplicity, in the following, we assume $\alpha_1 = \alpha_2 = 0$. In such a context, $\phi$ and $\psi$ satisfy

\begin{equation}\label{defphi}
\phi(K,L) = \begin{cases} K &\text{ if }  \frac{1}{K+\psi(K,L)+\epsilon} \geq c_1, \\ (c_1)^{-1} - \psi(K,L)-\epsilon &\text{ else},\end{cases}
\end{equation}
\begin{equation}\label{defpsi}
\psi(K,L) = \begin{cases} L &\text{ if }  \frac{1}{\phi(K,L)+L+\epsilon} \geq c_2, \\ (c_2)^{-1} - \phi(K,L)-\epsilon &\text{ else}.\end{cases}
\end{equation}

Let us note that the two previous equations do not properly define $\phi$ and $\psi$ in the whole domain $(\mathbb{R}_+)^2$, but only in the region of interest to define the maxima in (\ref{defU2})-(\ref{defV2}). Basically miners of type $i$ are all taking part in the mining activity if they can all gain from this, and if this is not the case, only a certain proportion of them is taking part in the mining industry. Obviously, when we are in the latter case, the proportion of miners of type $i$ which take part in the mining industry is such that all the miners of type $i$ are indifferent between mining and taking the reservation utility. Such a phenomenon is a Nash equilibria in mixed strategies in the game, we refer to \citeasnoun{bertucci2018optimal} and \citeasnoun{bertucci2018fokker} for more details on mixed Nash equilibria in MFG.  
The same approach for how to derive (\ref{masterEq}) can be used to derive the following system of 2 master equations.

\begin{equation}
\begin{dcases}
\begin{array}{@{}l}
\displaystyle
0=-(r_1+\delta)U+(-\delta K+\lambda_1 U)\partial_{K}U+(-\delta L+\lambda_2 V)\partial_{L}U \hspace{1cm}\\
\displaystyle
\hfill
+\max \left\{ \frac{1}{ \phi + \psi+ \epsilon}  -c_1; 0\right\}  \text{ in } [0,\infty)^2,\\
\end{array}
\\
\begin{array}{@{}l}
0=-(r_2+\delta)V+(-\delta K+\lambda_1 U)\partial_{K}V +(-\delta L+\lambda_2 V)\partial_{L}V \hspace{1cm}\\
\displaystyle
\hfill
+\max \left\{ \frac{1}{ \phi + \psi+ \epsilon}  -c_2; 0 \right\}  \text{ in } [0,\infty)^2,
\end{array}
\end{dcases}
\label{system}
\end{equation}

For the system (\ref{system}), we can obtain, as in the previous setting, the following result.

\begin{theorem}
There exists a unique lipschitz couple $(U,V)$ solution of (\ref{system}) such that $U(0,\cdot) \geq 0$ and $V(\cdot,0) \geq 0$. Moreover this couple satisfies the monotonicity property :
\begin{equation}
\begin{array}{l}
(U(x_1,y_1) - U(x_2,y_2))(x_1 - x_2) + (V(x_1,y_1) - V(x_2,y_2))(y_1 - y_2) \leq 0, \\
\forall x_1,x_2,y_1,y_2 \geq 0;
\end{array}
\end{equation}
together with the fact that both $U$ and $V$ are decreasing with respect to both $x$ and $y$.
\label{theorem2}
\end{theorem}

\begin{remark}
As in the one dimensional case, let us mention that no boundary condition is needed in this situation.
\end{remark}

\begin{proof}
The proof of this statement follows exactly the same argument as in the proof of theorem \ref{theorem2}. The monotonicity and uniqueness can be proven by looking at the equation satisfied by an auxiliary function $W$, which is in this case defined by
\[
W(x_1,y_1,x_2,y_2) = (U_1(x_1,y_1) - U_2(x_2,y_2))(x_1 - x_2) + (V_1(x_1,y_1) - V_2(x_2,y_2))(y_1 - y_2),
\]
for $(U_1,V_1)$ and $(U_2,V_2)$ two solutions of (\ref{system}). 
An a priori estimate which yields regularity and existence of solutions can be obtained by looking at the PDE satisfied by $W$ defined by
\[
\begin{array}{l}
W(x,y,\xi_1,\xi_2) = \partial_{K_1}U(x,y)\xi_1^2 + \partial_{K_2}V(x,y)\xi_2^2 + (\partial_{K_1}V + \partial_{K_2}U)\xi_1\xi_2 \hspace{1cm}\\
\hfill
+ \beta[(\partial_{K_1}U\xi_1)^2 + (\partial_{K_2}V\xi_2)^2],
\end{array}
\]
for some parameter $\beta$. Once the monotonicity property has been established, the fact that $U$ is decreasing with respect to $y$ (resp. $V$ is decreasing with respect to $x$) follows directly from the application of a maximum principle on the PDE satisfied by $\partial_y U$ (resp. $\partial_x V$).
\end{proof}

As in the simplest case we already treated, there exists a stationary state, and the induced trajectories converge toward this stationary state.

\begin{proposition}
There exists a unique stationary state $(x_0,y_0)$ and all the induced trajectories converge toward it.
\end{proposition}

\begin{proof}
Let us introduce some notations. We use the notation $(x,y) = z \in \mathbb{R}^2$ and we denote by $\langle \cdot, \cdot \rangle$ the scalar product on $\mathbb{R}^2$. We define the function $W : (\mathbb{R}_+)^2 \to \mathbb{R}^2$ by
\[
W(z) = (- \delta x + \lambda_1 U(z), -\delta y + \lambda_2 V(z)).
\]

From the monotonicity of $(U,V)$, we deduce that $W$ satisfies 
\[
\langle W(z_1) - W(z_2), z_1 - z_2 \rangle \leq -\delta |z|^2.
\]

From this property, existence and uniqueness of a stationary state $z_0$ (i.e. $z_0$ such that $W(z_0) = 0$) is standard. Moreover, to realize that a trajectory $(z(t))_{t \geq 0}$ converges toward $z_0$, one only need to compute the derivative in time of $|z(t) - z_0|^2$.
\end{proof}

Let us make 2 remarks. First, at the stationary state all existing machines are necessarily running. From the definition of the real hashrate dynamics (\ref{edo2pop}), if in a region in which some machines are turned off, miners of that type will stop buying new machines and therefore the stock of current machines will depreciate at the rate of technological progress until it reaches the stationary state. Second, the stationary state in the previous result can be located on the boundary of the domain $\{K \geq 0\}\cap \{L \geq 0\}$. For instance in the case $r_1=r_2$, $\lambda_1= \lambda_2$, $\alpha_1 = \alpha_2 = 0$ and $c_1 < c_2$, if 
\begin{equation}
K_* := \frac{\sqrt{(\delta \epsilon - \frac{c_1 \lambda_1}{r_1 + \delta})^2 + 4 \frac{\delta \lambda_1}{r_1 + \delta} } - \delta \epsilon - \frac{c_1\lambda_1}{r_1 + \delta}}{2 \delta}>\frac{1}{c_2} - \epsilon,
\end{equation}
then the only stationary equilibrium is $(K_*,0)$. This can easily be checked by remarking that $K_*$ is the equilibrium in the one population case if only the first miners are active. Thus if $c_2$ is large enough, at this level of total hashrate, miners from the second type have no interest in participating to mining industry. Thus $(K_*,0)$ is indeed a stationary equilibrium, and from the previous result, it is the sole one.

Otherwise stated, if the difference in electricity costs is large enough, only miners that have access to the cheap source of electricity will actively mine the blockchain in equilibrium. But if the difference in electricity price is small, even the miners that have access to the expensive source of electricity will be able to get a share of the real hashrate.

\subsection{Two populations of miners in a stochastic case}

We now present an extension with two different populations of miners, from two different countries, but we introduce some randomness that makes the relative profitability of miners in each population to vary. Once again, the case with more than two populations can be treated the same way. In some sense, it is an extension which put together the two previous cases. Each population faces the same cost structure but denominated in different currencies, one for each country, denoted by $Cur_i$ with $i\in\{1,2\}$. To make the analysis clearer we only consider variations in the exchange rate between the two currencies, denoted by $h(P_t)$ and assume no variation in the exchange rate between the reward of the blockchain and the cost. Without loss of generalities, we assume the reference currency is $Cur_1$. This means that the exchange rate between the cryptocurrency (the reward currency) and $Cur_1$ does not change and the exchange rate between the cryptocurrency and $Cur_2$ does change only through changes in $h(P_t)$.

The key idea for this extension to not be trivial is that the reward expressed in $Cur_2$ changes due to variations in the exchange rate because we assume the forex market (including the cryptocurrency) can be easily arbitraged, which is not the case for the electricity market\footnote{In practice there can be situations in which frictions on movement of fund from or to a specific country prevent full arbitrage. However we abstract from such frictions.}. Indeed, recall that the main cost a miner faces while computing hashes is electricity, and electricity prices are generally set locally. Since electricity is not storable and transfer of electricity between one corner of the world to an other can be very expensive (not to say impossible), the price of electricity will not adjust to variations in the exchange rate between the two national currencies\footnote{Very large players can have access to the wholesale market of electricity in which some adjustments may occur, but in general electricity prices will be driven by local supply and local demand.}. Note that we do not consider investment in energy production technologies and miners take the price of electricity as given.

The value function and the master equation will be written in the local currency for each country. In particular, this means that $c_1$ and $c_2$ are denominated in $Cur_1$ and $Cur_2$ respectively, and they are constant in each currency. Similarly to section \ref{randomPrice}, we assume that the exchange rate between the two national currencies is given by $h(P_t)$ where the dynamics of $(P_t)_{t\geq0}$ is given by (\ref{dynP}) and $h$ is a given positive smooth function. We keep up with the notations of the previous section.

In order to keep the same relative friction on the mining hardware market, the $\lambda$ for country $i=2$ also needs to be adjusted by $h(P_t)$, and we denoted it by $\lambda_2(h(P_t))$, while $\lambda_1$ denote the constant level of friction on the mining hardware market for country $i=1$.

Let's denote by $U$ and $V$ the value function for one unit of real hashrate in country $i=1$ and $i=2$ respectively. The analogous of (\ref{defU}) and (\ref{edo}) then becomes

\begin{equation}
U(K,L,p) := \mathbb{E}\left[ \int_0^{\infty} e^{-(r+\delta)t} \max\left(  \frac{1}{\phi(K_t,L_t,P_t)+\psi(K_t,L_t,P_t)+\epsilon} - c_1 ; 0 \right)  \right],
\end{equation}
\begin{equation}
V(K,L,p) := \mathbb{E}\left[ \int_0^{\infty} e^{-(r+\delta)t} \max\left(  \frac{h(P_t)}{\phi(K_t,L_t,P_t)+\psi(K_t,L_t,P_t)+\epsilon} - c_2 ; 0 \right)  \right],
\end{equation}
where $\phi$ and $\psi$ are defined in the analogue way as in the previous section and where $(K_t)_{t \geq 0}$, $(L_t)_{t\geq0}$ and $(P_t)_{t \geq 0}$ evolves according to 

\begin{equation}\label{edo2}
\begin{cases}
dP_t = \alpha(P_t)dt + \sqrt{2\nu}dW_t,\\
dK_t = - \delta K_t dt + \lambda_1 U(K_t,L_t,P_t)dt, \\
dL_t = - \delta L_t dt + \lambda_2(h(P_t)) V(K_t,L_t,P_t)dt,\\
K_0 = K , L_0=L , P_0 = p.
\end{cases}
\end{equation}
Using the same approach for how to derive (\ref{masterEq}) we can derive the following system of 2 master equations.

\begin{equation}
\begin{dcases}
\begin{array}{@{}l}
\displaystyle
0=-(r+\delta)U+(-\delta K+\lambda_1 U)\partial_{K}U+(-\delta L+\lambda_2(h(p)) V)\partial_{L}U\hspace{2cm}\\
\displaystyle
\hfill+\alpha\partial_pU+\nu\partial_{pp}U+\max\left\{\frac{1}{\phi + \psi + \epsilon}  -c_1 ; 0\right\}  \hspace{0.5cm}\text{ in } [0,\infty)^2\times\mathbb{R},\\
\end{array}
\\
\begin{array}{@{}l}
\displaystyle
0=-(r+\delta)V+(-\delta K+\lambda_1 U)\partial_{K}V +(-\delta L+\lambda_2(h(p)) V)\partial_{L}V\hspace{2cm}\\
\displaystyle
\hfill+\alpha\partial_pV+\nu\partial_{pp}V+\max\left\{\frac{h(p)}{\phi + \psi +\epsilon} -c_2 ;  0\right\} \hspace{0.5cm}\text{ in } [0,\infty)^2\times\mathbb{R},\\
\end{array}
\end{dcases}
\label{system}
\end{equation}

We leave results and proofs concerning this model to the interested reader as they can easily be established by combining the results in the previous models. In particular, due to the monotone structure, the mining game is still well posed in this context. Also, as in section \ref{randomPrice} the notion of stationary state does not really make sense, but for each value of the exchange rate $h(P_t)$, there exists a $Z_*(P_t)=(K_*(P_t),L_*(P_t))$ that will act like an attractor for the real hashrate in the two populations. And once again, because $U$ and $V$ are smooth, the function that maps $P_t$ into $Z_*(P_t)$ is also smooth.

\subsection{A different type of interaction on the product market}
\label{2p}

In this section we indicate how we can model different types of interactions for the market in which the machines are bought and sold. In our first model we consider a simple case in which this market is captured in the relation 

\begin{equation}
dK_t = - \delta K_t dt + \lambda U(K_t)dt.
\end{equation}

Let us mention that in general we could replace the term $\lambda U$ by a function $F(K,U)$. General assumptions for which results of well-posedness can be established on the master equation in this situation can be found in \citeasnoun{lions2007cours} \citeasnoun{bertucci2019some}. Here we want to focus on a particular model presented in \citeasnoun{prat2018equilibrium}. In their model, the machines cannot be sold and there is free entry in the market. This means that if the value of one unit of real hashrate is strictly positive, then machines are bought until this value becomes $0$. This behavior prevents the value $U$ of being strictly positive and thus imposes an obstacle on $U$. Following \citeasnoun{bertucci2018fokker}, \citeasnoun{bertucci2018optimal} and \citeasnoun{bertucci2020forthcomming}, it is natural to consider the following penalization for the trajectory

\begin{equation}
\label{uPosPart}
dK_t = - \delta K_t dt + \frac{1}{\eta} (U(K_t))_+dt,
\end{equation}
and to consider the limit $\eta \to 0$. Here $(\cdot)_+$ denotes the positive part. As $\eta$ goes to $0$, there is less and less friction in the hardware market and in the limit miners can freely buy as many machines as they want. Because of the positive part, the second term in (\ref{uPosPart}) cannot be negative and the real hashrate cannot decrease by more than the rate of technological progress. Note that this assumption is meant to represent how the mining market actually behaves for large blockchain. Indeed, on the largest blockchains, the hashing algorithm is so specifically designed into the hardware that mining machines can have a zero resale value.

Let us mention that a direct adaptation of the proof of theorem \ref{thm1} implies that for all $\eta> 0$, there is a unique lipschitz solution $U_{\eta}$ of

\begin{equation}\label{penal}
0=-(r+\delta)U+(-\delta K+\frac{1}{\eta} (U)_+)U'_K+\frac{1}{K+ \epsilon}-c \text{ in } [0,\infty).
\end{equation}

Moreover, checking the proof of theorem \ref{thm1}, the sequence $(U_{\eta})_{\eta > 0}$ is a sequence of uniformly lipschitz function. Hence, we can pass to the limit in (\ref{penal}) and obtain existence of a solution of 

\begin{equation}\label{obst}
\displaystyle
\max\left((r+\delta)U-(-\delta K)U'_K-\frac{1}{K+ \epsilon}+c, U \right) = 0\text{ in } [0, \infty),\\
\displaystyle
\end{equation}

\begin{theorem}
There exists a unique lipschitz and decreasing function $U$ solution of (\ref{obst}).
\end{theorem}
\begin{proof}
As already stated, the sequence $(U_{\eta})_{\eta > 0}$ is uniformly lipschitz. Evaluating the PDE satisfied by $U_{\eta}$ at $0$, we deduce that
\[
(r + \delta)U_{\eta}(0) = \frac{1}{\epsilon} - c + \frac{1}{\eta}(U_{\eta}(0))_+(U_{\eta})'_K(0).
\]
We then deduce that $(U_{\eta}(0))_{\eta}$ converges to $0$ as $\eta \to 0$. Thus from Ascoli-Arzela theorem, extracting a subsequence if necessary, $(U_{\eta})_{\eta >0}$ converges to some function $U$, which clearly solves (\ref{obst}). Let us now check that (\ref{obst}) has a unique solution. We denote by $U$ and $V$ two solutions of (\ref{obst}). We define
\[
W(x,y) = (U(x)-V(y))(x-y).
\]
We want to show that $W$ is negative. Let us assume that $W$ has a local maximum $(x_0,y_0)$ such that $W(x_0,y_0) > 0$. If $U(x_0) = 0$, then $x_0 > y_0$. Computing $\partial_yW(x_0,y_0)$, we obtain (because $V(y_0) < 0$) :
\[
\begin{aligned}
\partial_y W(x_0,y_0) &= V(y_0) + V'_K(y_0)(y_0 - x_0)\\
& = V(y_0) + (y_0 - x_0)\left( \frac{1}{y_0 + \epsilon} - c - (r + \delta)V(y_0)  \right)(\delta y_0)^{-1}\\
& < 0,
\end{aligned}
\]
where we have used that $1/(x_0 + \epsilon) > c$ and $x_0 < y_0$. This is a contradiction because we should have $\partial_y W (x_0,y_0) = 0$. We deduce that $U(x_0) < 0$. Similarly, we can assume that $V(y_0) < 0$. Thus using the PDE satisfied by $U$ and $V$, we deduce $W(x_0,y_0) \leq 0$. Hence, if $W$ has a local maximum it is negative. We leave to the interested reader the simple proof of the fact that $W(x,y)$ is negative when either $x$ or $y$ is large enough. We then conclude that $W \leq 0$ everywhere and thus we conclude as we did in the first model that there is a unique solution of (\ref{obst}).
\end{proof}

\begin{remark}
By combining this setting with the stochastic one, we can establish, following the same argument, uniqueness for the model of \citeasnoun{prat2018equilibrium}.
\end{remark}

We refer to \citeasnoun{bertucci2020forthcomming} for more results on this type of master equations. Once again, we further characterize the equilibrium by showing the existence of a stationary state for the real hashrate.

Concerning the trajectories associated with the limit model, there are given by $(K_t)_{t\geq 0}$ defined by 

\begin{equation}
\label{trajobst}
\begin{cases}
K_t = K_* \text{ for } t \geq 0 \text{ if } K_t \leq K_*,\\
dK_t = -\delta  K_t \text{ for } t \geq 0 \text{ if } K_t > K_*,
\end{cases}
\end{equation}
where $K_*$ is defined by $K_* := \inf\{K > 0 | U(K) < 0\}$ where $U$ is the unique solution of (\ref{obst}). Hence, as in the simpler model, there exists a unique stationary equilibrium $K_*$ and all the trajectories converge toward it. Moreover, the convergence is instantaneous if the initial condition, $K_0$, is below the equilibrium hashrate. Indeed, because of free entry in the mining industry, if the initial condition is such that $U>0$, that is there is some mining opportunity, miners will enter the game without restriction which will drive the value of mining toward $0$, and the equilibrium real hashrate will instantly reach the stationary state. Miners will keep entering the mining industry only to compensate for the technological progress, and the real hashrate will be constant. For completion, let's also detail the case for which the initial condition is such that $K_0>K_*$. Because in this setting miners cannot sell their mining machines nor turn them off, they will simply stop buying new machines and the real hashrate will decrease at the rate of technological progress until it reaches the stationary state. At this level, miners will start buying new machines again, but just to compensate for the decrease due to the technological progress.

Moreover, to justify the form of the trajectories of (\ref{trajobst}), we can establish a result of convergence of trajectories.

\begin{proposition}
\label{trajLimit}
For any $K_0 \geq 0$, for any $\eta > 0$ denote by $(K^{\eta}_t)_{t\geq 0}$ the solution of 
\[
\begin{cases}
dK^{\eta}_t = - \delta K^{\eta}_t dt + \frac{1}{\eta} (U_{\eta}(K^{\eta}_t))_+dt,\\
K^{\eta}_0 = K_0,
\end{cases}
\]
where $U_{\eta}$ is the unique solution of (\ref{penal}). Denote by $(K_t)_{t \geq 0}$ the solution of (\ref{trajobst}). Then, for any $t> 0$, 
\[
\lim_{\eta \to 0} K^{\eta}_t = K_t.
\]
\end{proposition}

\begin{proof}
For any $\eta> 0$, from the decreasing property of $U^{\eta}$, either $(K^{\eta}_t)_{t\geq 0}$ is decreasing or it is increasing and it converges toward a stationary state $K_*^{\eta}$ which is bounded (in $\eta$). Thus as $\eta \to 0$, $(K^{\eta}_t)_{t\geq 0}$ converges simply toward a limit trajectory $(K'_t)_{t\geq 0}$ (not necessary continuous but with the same monotony). For all $t$ such that $U(K'_t)< 0$, because $\{U< 0\}$ is an open set and $(U^{\eta})_{\eta> 0}$ converges uniformly toward $U$, we deduce that
\begin{equation}\label{conve}
\lim_{\eta \to 0} \frac{1}{\eta}U_{\eta}(K^{\eta}_t) = U(K'_t) = 0.
\end{equation}
Let us now consider $t > 0$ such that $U(K'_t) = 0$.
\[
K^{\eta}_t - K_0 = \int_0^{t} (-\delta K^{\eta}_s+ \frac{1}{\eta}(U_{\eta}(K^{\eta}_s))_+)ds.
\]
The monotone convergence theorem implies that
\begin{equation}\label{convmon}
\underset{\eta \to 0}{\text{limsup}} \int_0^{t} \frac{1}{\eta}(U_{\eta}(K^{\eta}_s))_+ds < \infty.
\end{equation}
However, from the convergence of $(U_{\eta})_{\eta> 0}$ toward $U$ established in the previous result, we known that for any $K$ in the interior of $\{U = 0\}$, 
\[
\lim_{\eta \to 0}\frac{1}{\eta}(U_{\eta}(K))_+ = + \infty.
\]
Thus we deduce from (\ref{convmon}) that for almost every $s \in (0,t)$, 
\begin{equation}\label{convh}
\lim_{\eta \to 0}K^{\eta}_t = K_*:=\inf\{K > 0 | U(K) < 0\}.
\end{equation}
From (\ref{conve}) and (\ref{convh}) we easily deduce that $(K'_t)_{t \geq 0} = (K_t)_{t \geq 0}$.
\end{proof}

\begin{remark}
\label{statStatePostPart}
Let us remark that in particular, the stationary state of the penalized problem converges toward the stationary state $K_*$ of the limit problem. By realizing that replacing $\lambda$ by $\eta^{-1}$ in (\ref{equstat}) we obtain the stationary state of the penalized problem, we can compute $K_*$ by passing to the limit $\lambda \to \infty$ in (\ref{equstat}).
\end{remark}

It turns out that the stationary state as defined in (\ref{equstat}) is the same as when the dynamics of the real hashrate is given by (\ref{uPosPart}). In other words, in equilibrium, the real hashrate reaches the same stationary state whether miners can sell their machines or not. This arises because for $\eta>0$, the value of the real hashrate at the stationary state is positive, $U_*>0$, so taking the positive part does not change the equilibrium. In economics terms, because the real hashrate constantly depreciates at the rate of technological progress, miners do not need to be allowed to sell their machines, the real hashrate will still decrease if they do not buy any. The speed of convergence will be affected of course, but the stationary state will eventually reach the same value.

Note that the stationary state in this case correspond to the limit case if we were to continue the graph in figure \ref{comparativeLambda} on the x-axis. In particular we can compute the limit stationary state when there are no frictions, $K_*$. By proposition \ref{trajLimit} we know the penalized state tends to the limit case, as $\lambda\rightarrow+\infty$. So, by taking the limit of (\ref{equstat}), we obtain
\begin{equation}
K_*=\frac{1}{c}-\epsilon.
\end{equation}

\section{Comment on the  modeling and comparaison with standard HJB equation}
\label{masterVsHJB}

We now discuss on the choices of modeling we made. We claim that looking at the master equation for this kind of problem is probably the most general approach one can have. Indeed looking at the master equation allows us a great liberty in our model, while keeping the main structure of the equations (well-posedness, regularity, monotone structure...). We believe the previous series of extensions can convinced the reader of this fact. The main advantage of our point of view is the one of the MFG theory, which is that we can explain the behavior of aggregate quantities (here the hashrate $(K_t)_{t \geq 0}$) from microscopic decisions (here buying or selling machines).\\

To highlight this claim, let us compare our approach with a more traditional one in Economics which would have been to consider the problem of a social planner. The link between the social planner problem and the MFG equilibrium is a well understood question in the MFG literature, see for instance \citeasnoun{lasry2007mean}, \citeasnoun{cardaliaguet2019efficiency}. The case in which the two problems can lead to the same outcome is called, as usual in the game literature, the potential case. In this potential case, the solution of the master equation derives from the gradient of a solution of a Hamilton-Jacobi-Bellman (HJB) equation.\\

Our first model is a potential case and the solution $U$ of the master equation (\ref{masterEq}) is the derivative of the function $\Phi$ solution of the HJB equation 
\begin{equation}\label{hjb}
0=-r\Phi-\delta K\Phi+\frac{\lambda}{2} (\Phi'_K)^2+ln(K+ \epsilon)-cK.
\end{equation}

with state constraint at $K = 0$. The study of (\ref{hjb}) is somehow more classical and easy, and furthermore, the function $\Phi$ has a nice interpretation in the model, it is the value function a monopolist would face if it was the sole owner of the fleet of machines. This HJB equation is obtained by integrating the master equation with respect to $K$.

However, in general, there is no reason why the behavior of a market composed of an infinite number of small agents could be explain by the decision of a single agent. In fact one should consider as the exception and not the rule that we are in the potential case. For instance the more complex case developed in section \ref{2p} is not a potential case. That is why we believe that it is important to develop an approach which is robust to small changes in the model, which is the case of the one we presented. Let us mention, in the setting of the various models we described, some assumptions under which we are not in the potential case.
\begin{itemize}
\item
If $\lambda$ depends on $K$.
\item
In the two dimensional case (i.e. when there is $K$ and $L$), when the derivative of the reward of the first type of miners with respect to $L$ is different than the derivative of the reward of the second type of miners with respect to $K$.
\item
In the two dimensional case, if there are two different discount rates $r_1$ and $r_2$.
\end{itemize}

Let us finally comment on the fact that in a one-dimensional case, because it is quite easy to integrate an equation in only one variable, we are in some sense more frequently in the potential case.

\section{Conclusion}
\label{conclusion}

In this paper, we have built a flexible framework for thinking about the proof-of-work distributed consensus algorithm. In particular, we have analyzed the dynamics of the total hashrate which represents the total computational power devoted to mining for the blockchain.

For the sake of clarity, we have presented our model in increasing complexity. In the base case setting, there is a single infinite population of miners, the price of the cryptocurrency is constant, and miners cannot turn off the computing chips. We show existence and uniqueness of equilibrium. In equilibrium, the hashrate follows the rate of technological progress. 

We then provide a set of extensions that allow us to extend the result of the baseline model. We introduce a random exchange rate, several populations of miners facing different mining cost, as well as different kinds of interaction in the market for mining hardware. In all of those extensions, the essential structure of the problem is preserved and we prove existence and uniqueness of equilibrium.

There is a direct connection between the hashrate, the security, and the energy consumption on a PoW chain. Indeed, the most common attack is the 51\% attack, which describes an attacker controlling the majority of the hashrate that would allow him to make a double spend, or even arbitrarily rewrite part of the blockchain. The incentive scheme has been designed to minimize the likelihood of such an attack by maximizing the cost of a 51\% attack. Every thing else being equal, the more the hashrate, the more it is expensive to control half of it. The security of the blockchain is directly linked to the cost of computing the total hashrate. Therefore the fact that in equilibrium the hashrate follows the rate of technological progress imply that the security of the blockchain is constant. An other way to put it is to say that the energy consumption is constant in equilibrium. Indeed, the quantity of energy needed to compute the hashrate (the cost of the hashrate) is constant and the only increase in hashrate comes from the technological progress.

\section*{Acknowledgments}
The first, third and fourth authors have been partially supported by the Chair FDD (Institut Louis Bachelier). The fourth author has been partially supported by the Air Force Office for Scientific Research grant FA9550-18-1-0494 and the Office for Naval Research grant N000141712095.

\bibliographystyle{jf}
\bibliography{bib}

\end{document}